\documentclass[lettersize,journal]{IEEEtran}

\usepackage{amsmath,amssymb,amsthm,mathtools}
\usepackage{bm}
\usepackage{bbm}
\usepackage{graphicx}
\usepackage{algorithm}
\usepackage{algpseudocode}
\usepackage{booktabs}
\usepackage{cite}
\usepackage{url}
\usepackage[hidelinks]{hyperref}

\newtheorem{theorem}{Theorem}
\newtheorem{lemma}{Lemma}

\newtheorem{corollary}{Corollary}
\theoremstyle{remark}
\newtheorem{remark}{Remark}
\newtheorem{assumption}{Assumption}
\newcommand{\ind}{\mathbbm{1}}
\newcommand{\er}{\mathsf{e}}
\newcommand{\sigmoid}{\sigma}

\begin{document}

\title{Distributed Learning  over Noisy Communication Networks}

\author{ Emrah Akyol  \IEEEmembership{Senior Member, IEEE}  and Marcos M. Vasconcelos  \IEEEmembership{Member, IEEE} 

\thanks{EA is with the ECE Department, Binghamton University, SUNY, NY, USA. MV is with the ECE Department, Florida State University, FL, USA. emails: eakyol@binghamton.edu, marcos@eng.famu.fsu.edu.} \thanks{This work is supported in part by the NSF via CCF/CIF (CAREER) \#2048042. }
}

\maketitle

\begin{abstract}
We study binary coordination games over graphs under log-linear learning when neighbor actions are conveyed through explicit noisy communication links. Each edge is modeled as either a binary symmetric channel (BSC) or a binary erasure channel (BEC). We analyze two operational regimes.  For binary
symmetric and binary erasure channels, we provide a structural characterization of the
induced learning dynamics. In a fast-communication regime, agents update using channel-averaged payoffs; the resulting learning dynamics coincide with a Gibbs sampler for a scaled coordination potential, where channel reliability enters only through a scalar attenuation coefficient. In a snapshot regime, agents update from a single noisy realization and ignore channel statistics; the induced Markov chain is generally nonreversible, but admits a high-temperature expansion whose drift matches that of the fast Gibbs sampler with the same attenuation. We further formalize a finite-$K$ communication budget, which interpolates between snapshot and fast behavior as the number of channel uses per update grows. This viewpoint yields a communication-theoretic interpretation in terms of retransmissions and repetition coding, and extends naturally to heterogeneous link reliabilities via effective edge weights. Numerical experiments illustrate the theory and quantify the tradeoff between communication resources and steady-state coordination quality.
\end{abstract}


\section{Introduction}
\label{sec:intro}

Large-scale networked systems are now a cornerstone of modern communication, computation, and control infrastructures. Applications ranging from machine learning and distributed sensing to cooperative robotics, smart grids, and Internet-of-Things (IoT) platforms increasingly rely on collections of agents that coordinate through local message exchanges over communication networks. In these systems, global objectives, such as consensus, distributed optimization, collaborative inference, or learning a shared model, must be achieved through iterative local interactions constrained by the physical properties of the underlying communication medium.

From a communications perspective, a defining feature of such distributed cooperative systems is that information exchange is inherently imperfect. Bandwidth limitations, noise, quantization, packet losses, interference, and intermittent connectivity are fundamental constraints in wireless and large-scale wired networks. These impairments are typically addressed at the PHY, MAC, and network layers through coding, scheduling, and retransmission mechanisms. However, even when such mechanisms are present, higher-layer distributed algorithms operate on information that is delayed, noisy, or partially observed. As a result, communication constraints ultimately shape the collective performance of distributed learning and coordination protocols.

This paper studies distributed coordination through the lens of {log-linear learning} (LLL) over communication-constrained networks. LLL and related stochastic best-response dynamics are widely used abstractions for decentralized adaptation in networked systems, with connections to distributed optimization, consensus algorithms, and learning in games. Their appeal lies in their simplicity and robustness: agents update using local information and randomized rules that naturally tolerate uncertainty. The existing literature, however, typically assumes that agents observe their neighbors’ actions or states reliably, or models communication imperfections indirectly as abstract noise in their utility functions \cite{MardenShamma2012} or random link failures \cite{MuralidharanMostofiMILCOM2015}.

In contrast, we explicitly model the physical communication layer. We focus on binary coordination games over graphs, which serve as a minimal yet expressive abstraction for  cooperative decision-making in distributed systems. Agents exchange information with neighbors through noisy communication channels, such as binary symmetric or erasure channels, and base their updates on the received messages. This explicit channel-level modeling allows us to directly connect communication reliability to effective inter-agent coupling strengths, providing a principled bridge between communication theory and distributed learning dynamics.

\subsection{Contributions}
\label{subsec:contrib}

 \emph{Exact Gibbs structure in the fast regime:} For BSC and BEC links, fast communication turns LLL into a single-site Gibbs sampler
for a scaled coordination potential. Communication enters only through a scalar attenuation
coefficient $\kappa$ defined in \eqref{eq:kappa_def}.

\emph{Snapshot regime as a nonequilibrium process with a controlled high-temperature limit:}
Snapshot updates generate a generally nonreversible Markov chain.
For small inverse temperature $\beta$, we derive a first-order expansion
that recovers the same drift as the fast Gibbs sampler, with the same attenuation \eqref{eq:kappa_def}.

 \emph{Finite-$K$ communication budget and interpolation:}
We formalize an intermediate model in which each update uses $K$ channel uses per neighbor.
We prove entrywise kernel convergence to the fast regime as $K\to\infty$
and convergence of stationary distributions on the finite state space.

 \emph{Communication-theoretic interpretation and heterogeneity:}
The finite-$K$ model yields a concrete interpretation in terms of retransmissions and
repetition coding. We also extend the fast regime to heterogeneous link qualities,
showing that reliability simply rescales edge weights in the Gibbs potential.

Overall, this work contributes to the understanding of distributed coordination over communication-constrained networks by providing a unified analytical framework that connects channel models, network topology, and learning dynamics. By making the role of communication explicit, our results complement existing research on communication-efficient distributed optimization and learning, offers new insights into designing novel communication-aware, scalable, and resilient distributed cooperative systems. 

\subsection{Related Work}
\label{subsec:related}

The modern use of LLL in games traces to Blume~\cite{Blume1993}, who made explicit the Gibbs connection for potential games. Subsequent work developed robustness, implementability, and structural relaxations; see, e.g., Marden and Shamma's revisiting of LLL under asynchrony and payoff-based feedback~\cite{MardenShamma2012}. In parallel, the algorithmic game theory community analyzed logit dynamics beyond potential maximizers, including bounds on mixing time and metastability; see Auletta et al.\ for general mixing bounds and social welfare under logit dynamics~\cite{AulettaJACM2011,AulettaSAGT2010} and metastability phenomena in coordination games~\cite{AulettaSODA2012}. Classical selection under perturbations via evolutionary logit response and mutations goes back to Young and to Kandori--Mailath--Rob~\cite{Young1993,KMR1993}.

There is a substantial literature showing that LLL is robust to \emph{utility} noise and limited feedback, including finite-time convergence to $\varepsilon$-efficient Nash equilibria in general potential games (under Lipschitz regularity and noisy/partial feedback)~\cite{Maddux2024}. There is also work on {stochastic connectivity}, e.g., BLLL over randomly available links,  that derives conditions on link-up probabilities for concentrating the stationary measure on potential maximizers~\cite{MuralidharanMostofiMILCOM2015}. However, to our knowledge, these strands do not model a {per-edge observation channel} of neighbor actions and drive the logit via the associated {channel likelihood/posterior}. Our formulation fills this gap by (i) placing a BSC$(p)$ or BEC($\varepsilon$) on each link, (and (ii) showing a two-regime attenuation law, $\kappa =(1-2p)$ for BSC and $\kappa =(1-\varepsilon)$ for BEC  in the averaged (fast-channel) regime and in the snapshot regime to first order in~$\beta$.

\section{Model and Preliminaries}
\label{sec:model}

\subsection{Coordination game and potential}

Let $G=(V,E)$ be an undirected graph with $|V|=n$ and symmetric nonnegative weights
$v_{ij}=v_{ji}$ for $\{i,j\}\in E$.

Each agent $i$ selects an action $x_i\in\{0,1\}$, and the joint action profile is $x=(x_1,\dots,x_n)$. We consider a binary coordination game in which agents are rewarded for matching the actions of their neighbors. The utility of agent $i$ is
\[
u_i(x)=\sum_{j\in N(i)} v_{ij}\,\mathbf{1}\{x_i=x_j\}.
\]
This game is a potential game with potential function
\begin{equation}
\Phi(x)
\triangleq
\sum_{\{i,j\}\in E} v_{ij}\,\ind\{x_i=x_j\}.
\label{eq:Phi_def}
\end{equation}
The two consensus profiles $x=\mathbf{0}$ and $x=\mathbf{1}$ are Nash equilibria and represent global coordination.

For each $i$, we define 
\begin{equation}
m_i(b;x)
\triangleq
\sum_{j\in N(i)} v_{ij}\,\ind\{x_j=b\},
\qquad b\in\{0,1\},
\label{eq:mi_def}
\end{equation}
where $N(i)\triangleq \{j:\{i,j\}\in E\}$.

Only edges incident to $i$ change when $x_i$ changes. In particular,
for any $b,b'\in\{0,1\}$,
\begin{equation}
\Phi(b,x_{-i})-\Phi(b',x_{-i})
=
m_i(b;x)-m_i(b';x).
\label{eq:Phi_local_diff}
\end{equation}

We also define the single-site potential difference
\begin{equation}
\Delta\Phi_i(x)
\triangleq
\Phi(1,x_{-i})-\Phi(0,x_{-i}),
\label{eq:DeltaPhi_def}
\end{equation}
which is equivalen tto, via \eqref{eq:Phi_local_diff}: \begin{equation} \label{rel}\Delta\Phi_i(x)=m_i(1;x)-m_i(0;x). \end{equation}

 In exact-potential coordination games (such as the one at hand), the alignment between individual incentives and a global potential $\Phi(x)$ makes equilibrium selection naturally interpretable as an instance of entropy-regularized optimization. In particular, consider the {free-energy} maximization over distributions $\mu$ on the joint-action space,
\begin{equation}
\max_{\mu \in \Delta(\mathcal{X})}
\;\;
\mathbb{E}_{\mu}\!\left[\Phi(x)\right]
+
\frac{1}{\beta}\, H(\mu),
\label{eq:free_energy}
\end{equation}
where $H(\mu)\triangleq -\sum_{a\in\mathcal{X}} \mu(x)\log \mu(x)$ is the Shannon entropy and $\beta>0$ is an inverse-temperature parameter controlling the exploration-exploitation tradeoff. The unique optimizer of~\eqref{eq:free_energy} is the Gibbs distribution
\begin{equation}
\mu^\star(x)
=
\frac{\exp\!\big(\beta \Phi(x)\big)}{\sum_{x'\in\mathcal{X}}\exp\!\big(\beta \Phi(x')\big)},
\qquad x\in\mathcal{X},
\label{eq:gibbs}
\end{equation}
i.e., the maximum-entropy distribution consistent with a soft preference for high-potential configurations. Log-linear learning (logit response) provides a fully decentralized mechanism that realizes~\eqref{eq:gibbs}: when agents update asynchronously by sampling actions according to
\begin{equation}
\Pr\!\big(x_i \mid x_{-i}\big)
=
\frac{\exp\!\big(\beta\, u_i(x_i,x_{-i})\big)}
{\sum_{x_i' \in \mathcal{X}_i}\exp\!\big(\beta\, u_i(x_i',x_{-i})\big)},
\label{eq:logit}
\end{equation}
the induced Markov chain is reversible and admits $\mu^\star$ as its stationary distribution whenever $\Phi$ is a potential for the utilities $\{u_i\}$. Hence, log-linear learning can be viewed as a distributed Gibbs sampler on $\mathcal{X}$: each local update implements a conditional Gibbs step, and the temperature $\beta^{-1}$ tunes the degree of randomization. This equivalence links coordination-game learning dynamics to entropy-regularized (smoothed) combinatorial optimization, with the limiting regime $\beta\to\infty$ concentrating $\mu^\star$ on global maximizers of $\Phi$, while finite $\beta$ yields controlled exploration and robustness in the presence of noise and model mismatch. This connection is made explicity in Section \ref{sec:variational} to interpret our results for distributed optimization and learning problems. 

\subsection{Communication model}
\label{subsec:channel_model}

When agent $i$ updates, it receives from each neighbor $j\in N(i)$
a symbol $y_{ij}$ generated from $x_j$ through a memoryless channel.

\subsubsection{Binary symmetric channel}

Each link $\{i,j\}$ is modeled as a binary symmetric channel (BSC) with crossover probability $p\in[0,1/2)$. Agent $i$ receives $y_{ij}=x_j$ with probability $1-p$ and $y_{ij}=1-x_j$ with probability $p$. More formally, for $p\in[0,\tfrac12)$,
\begin{equation}
y_{ij}
=
x_j\oplus \eta_{ij},
\label{eq:bsc_model}
\end{equation}
where $\eta_{ij}\sim \mathrm{Bernoulli}(p)$ is independent across edges and update times. 

\subsubsection{Binary erasure channel}
Alternatively, communication may occur over a binary erasure channel (BEC) with erasure probability $\varepsilon\in[0,1)$. Agent $i$ receives $y_{ij}=x_j$ with probability $1-\varepsilon$ and an erasure symbol $\er$ with probability $\varepsilon$, ie.,  for $\epsilon\in[0,1)$,
\begin{equation}
y_{ij}
=
\begin{cases}
x_j, & \text{with probability }1-\epsilon,\\
\er, & \text{with probability }\epsilon.
\end{cases}
\label{eq:bec_model}
\end{equation}

In the BEC case, erased symbols contribute zero payoff, i.e.,
$\ind\{\er=b\}=0$ for $b\in\{0,1\}$.

Our analysis show that both channels can be characterized via a unified scalar attenuation coefficient that we will refer to throughout the paper: 
\begin{equation}
\kappa
\triangleq
\begin{cases}
1-2p, & \text{for $\mathrm{BSC}(p)$},\\
1-\epsilon, & \text{for $\mathrm{BEC}(\epsilon)$}.
\end{cases}
\label{eq:kappa_def}
\end{equation}

\subsection{Asynchronous log-linear learning}
\label{subsec:lll}

Agents update asynchronously using log linear learning. At each time step, one agent $i$ is selected uniformly at random and chooses action $b\in\{0,1\}$ with probability

\begin{equation}
\mathbb{P}(x_i^+=b)
=
\frac{\exp(\beta\,\widehat u_i(b))}
{\sum_{b'\in\{0,1\}}\exp(\beta\,\widehat u_i(b'))}.
\label{eq:logit_general}
\end{equation}
where  $\widehat u_i(b)$ is teh payoff estimate for $b\in\{0,1\}$. This is equivalent to a logit update with inverse temperature $\beta>0$, where the logistic function is: 
\begin{equation}
\sigmoid(t)
\triangleq
\frac{1}{1+e^{-t}}.
\label{eq:sigmoid_def}
\end{equation}

\section{Snapshot and Fast Communication Regimes}
\label{sec:regimes}

\subsection{Snapshot regime}
\label{subsec:snapshot}

In the snapshot regime, agent $i$ updates using a single noisy observation
$y_i=(y_{ij})_{j\in N(i)}$ and treats it as the true neighbor action profile. Here each agent update is based on a single realization of the communication process. i.e., agents do not average over channel noise and instead respond to instantaneous noisy observations. 

The snapshot payoff estimate is agent $i$ taking action $b$ is
\begin{equation}
\widehat u_i^{\mathrm{S}}(b\mid y_i)
\triangleq
\sum_{j\in N(i)} v_{ij}\,\ind\{y_{ij}=b\},
\qquad b\in\{0,1\}.
\label{eq:uhat_snap_def}
\end{equation}

The following expression that represents the difference in payoffs (advantage) will be used throughout the paper: 
\begin{equation}
\widehat\Delta_i^{\mathrm{S}}(y_i)
\triangleq
\widehat u_i^{\mathrm{S}}(1\mid y_i)-\widehat u_i^{\mathrm{S}}(0\mid y_i).
\label{eq:Delta_snap_def}
\end{equation}

Using \eqref{eq:logit_general}, the update depends only on the advantage \eqref{eq:Delta_snap_def},
\begin{equation}
\mathbb{P}(x_i^+=1\mid y_i)
=
\sigmoid\!\left(\beta\,\widehat\Delta_i^{\mathrm{S}}(y_i)\right).
\label{eq:snapshot_logit_1}
\end{equation}
\begin{equation}
\mathbb{P}(x_i^+=0\mid y_i)
=
\sigmoid\!\left(-\beta\,\widehat\Delta_i^{\mathrm{S}}(y_i)\right).
\label{eq:snapshot_logit_0}
\end{equation}

Averaging over the channel and the random node selection yields a Markov kernel $P_\beta^{\mathrm{S}}$ on $\{0,1\}^n$.
For $x_i=0$,
\begin{equation}
P_\beta^{\mathrm{S}}\!\left(x,x^{(i,1)}\right)
=
\frac{1}{n}\,
\mathbb{E}\!\left[
\sigmoid\!\left(\beta\,\widehat\Delta_i^{\mathrm{S}}(y_i)\right)
\,\middle|\,
x
\right].
\label{eq:P_snap_forward}
\end{equation}
For $x_i=1$,
\begin{equation}
P_\beta^{\mathrm{S}}\!\left(x,x^{(i,0)}\right)
=
\frac{1}{n}\,
\mathbb{E}\!\left[
\sigmoid\!\left(-\beta\,\widehat\Delta_i^{\mathrm{S}}(y_i)\right)
\,\middle|\,
x
\right].
\label{eq:P_snap_backward}
\end{equation}
It is unclear at first sight that such dynamics induce a unique stationary distribution. The following theorem states broadly states the ergodicity of the snapshot dynamics, hence its existence and  properties.

\begin{theorem}
\label{thm:snap_stationary}
The Markov chain with transition matrix $P_\beta^{\mathrm{S}}$ is irreducible and aperiodic.
Hence it admits a unique stationary distribution $\pi_\beta^{\mathrm{S}}$.
\end{theorem}

\begin{proof}
For every $y_i$ and $\beta>0$, \eqref{eq:snapshot_logit_1}-\eqref{eq:snapshot_logit_0} imply
$\mathbb{P}(x_i^+=0\mid y_i)\in(0,1)$ and $\mathbb{P}(x_i^+=1\mid y_i)\in(0,1)$.
Thus, for every $x$ and every $i$, the probability of setting $x_i$ to either value is strictly positive
after averaging over $y_i$.

From any $x$ to any $x'$ one can reach $x'$ by finitely many single-coordinate changes.
Each such single-coordinate change occurs with positive probability under $P_\beta^{\mathrm{S}}$ as shown above. 
Hence the chain is irreducible.

Fix $x$. Select any $i$. With positive probability the update chooses $x_i^+=x_i$,
hence $P_\beta^{\mathrm{S}}(x,x)>0$ and the chain is aperiodic.

A finite irreducible aperiodic Markov chain has a unique stationary distribution.
\end{proof}

\subsubsection{High-temperature expansion}
We first present this auxilliary result.
\begin{lemma}
\label{lem:mean_snap_adv}
For either channel model, the mean snapshot advantage satisfies
\begin{equation}
\mathbb{E}\!\left[\widehat\Delta_i^{\mathrm{S}}(y_i)\,\middle|\,x\right]
=
\kappa\,\Delta\Phi_i(x),
\label{eq:E_Delta_snap}
\end{equation}
where $\kappa$ is defined in \eqref{eq:kappa_def}.
\end{lemma}
\begin{proof}
By \eqref{eq:uhat_snap_def} and \eqref{eq:Delta_snap_def},
\begin{equation}
\widehat\Delta_i^{\mathrm{S}}(y_i)
=
\sum_{j\in N(i)} v_{ij}
\Big(\ind\{y_{ij}=1\}-\ind\{y_{ij}=0\}\Big).
\label{eq:Delta_expand_app}
\end{equation}
Taking conditional expectation given $x$ and using linearity,
\begin{equation}
\mathbb{E}\!\left[\widehat\Delta_i^{\mathrm{S}}(y_i)\,\middle|\,x\right]
=
\sum_{j\in N(i)} v_{ij}\,
\mathbb{E}\!\left[\ind\{y_{ij}=1\}-\ind\{y_{ij}=0\}\,\middle|\,x_j\right].
\label{eq:Delta_mean_app}
\end{equation}

For $\mathrm{BSC}(p)$, the inner expectation equals $(1-2p)$ when $x_j=1$ and equals $-(1-2p)$ when $x_j=0$.
For $\mathrm{BEC}(\epsilon)$ (with $\ind\{\er=b\}=0$), it equals $(1-\epsilon)$ when $x_j=1$ and equals $-(1-\epsilon)$ when $x_j=0$.
In both cases,
\begin{equation}
\mathbb{E}\!\left[\ind\{y_{ij}=1\}-\ind\{y_{ij}=0\}\,\middle|\,x_j\right]
\!=\!
\kappa\Big(\ind\{x_j=1\}-\ind\{x_j=0\}\Big)\!
\label{eq:inner_mean_app}
\end{equation}
where $\kappa$ is given by \eqref{eq:kappa_def}.
Substituting \eqref{eq:inner_mean_app} into \eqref{eq:Delta_mean_app} yields
\begin{equation}
\mathbb{E}\!\left[\widehat\Delta_i^{\mathrm{S}}(y_i)\,\middle|\,x\right]
=
\kappa\big(m_i(1;x)-m_i(0;x)\big).
\label{eq:E_Delta_app}
\end{equation}
Using $\Delta\Phi_i(x)=m_i(1;x)-m_i(0;x)$ from \eqref{eq:DeltaPhi_def} completes the proof.
\end{proof}

\begin{theorem}
\label{thm:2}
Fix $x\in\{0,1\}^n$ and $i\in V$. As $\beta\to 0$,
\begin{equation}
\mathbb{P}(x_i^+=1\mid x)
=
\frac{1}{2}
+
\frac{\beta}{4}\,
\kappa\,\Delta\Phi_i(x)
+
o(\beta).
\label{eq:snapshot_smallbeta_prob}
\end{equation}
Equivalently, for $x_i=0$,
\begin{equation}
P_\beta^{\mathrm{S}}\!\left(x,x^{(i,1)}\right)
=
\frac{1}{2n}
+
\frac{\beta}{4n}\,
\kappa\,\Delta\Phi_i(x)
+
o(\beta).
\label{eq:snapshot_smallbeta_kernel}
\end{equation}
\end{theorem}
\begin{proof}
Conditioning on $x$ and using \eqref{eq:snapshot_logit_1},
\begin{equation}
\mathbb{P}(x_i^+=1\mid x)
=
\mathbb{E}\!\left[\sigmoid\!\left(\beta\,\widehat\Delta_i^{\mathrm{S}}(y_i)\right)\,\middle|\,x\right].
\label{eq:snap_prob_app}
\end{equation}
Since $\sigmoid(0)=1/2$ and $\sigmoid'(0)=1/4$,
\begin{equation}
\sigmoid(t)
=
\frac12+\frac{t}{4}+o(t),
\qquad t\to 0.
\label{eq:sigmoid_expansion_app}
\end{equation}
Substituting \eqref{eq:sigmoid_expansion_app} into \eqref{eq:snap_prob_app} yields
\begin{equation}
\mathbb{P}(x_i^+=1\mid x)
=
\frac12+\frac{\beta}{4}\,
\mathbb{E}\!\left[\widehat\Delta_i^{\mathrm{S}}(y_i)\,\middle|\,x\right]
+o(\beta).
\label{eq:snap_prob_expand_app}
\end{equation}
Lemma~\ref{lem:mean_snap_adv} provides \eqref{eq:E_Delta_snap}, giving \eqref{eq:snapshot_smallbeta_prob}.
Multiplying by $1/n$ yields \eqref{eq:snapshot_smallbeta_kernel}.
\end{proof}
\begin{corollary}
\label{cor:gibbs_like_first_order_snap}
Define the Gibbs candidate
\begin{equation}
\widetilde{\pi}_\beta(x)\triangleq
\frac{\exp\!\big(\beta \kappa \Phi(x)\big)}
{\sum_{z\in\{0,1\}^n}\exp\!\big(\beta \kappa \Phi(z)\big)}.
\label{eq:gibbs_candidate_snap_unaware}
\end{equation}
Fix an adjacent pair $x$ and $x^{(i,1)}$ with $x_i=0$.
Then, as $\beta\to 0$,
\begin{equation}
\log\frac{P_\beta^{\mathrm{S}}(x,x^{(i,1)})}{P_\beta^{\mathrm{S}}(x^{(i,1)},x)}
=
\beta \kappa \big(\Phi(x^{(i,1)})-\Phi(x)\big)+o(\beta),
\label{eq:logratio_first_order_snap}
\end{equation}
and consequently
\begin{equation}
\log\frac{\widetilde{\pi}_\beta(x)\,P_\beta^{\mathrm{S}}(x,x^{(i,1)})}
{\widetilde{\pi}_\beta(x^{(i,1)})\,P_\beta^{\mathrm{S}}(x^{(i,1)},x)}
=
o(\beta).
\label{eq:approx_detailed_balance_first_order}
\end{equation}
\end{corollary}

\begin{proof}
Let $x_i=0$ and set $x'=x^{(i,1)}$.
By Theorem~\ref{thm:2},
\begin{align*}
P_\beta^{\mathrm{S}}(x,x')&=\frac{1}{2n}+\frac{\beta}{4n} \kappa \Delta\Phi_i(x)+o(\beta),
\\
P_\beta^{\mathrm{S}}(x',x)&=\frac{1}{2n}-\frac{\beta}{4n}\kappa \Delta\Phi_i(x)+o(\beta),
\end{align*}
since $\Phi(x')-\Phi(x)=\Delta\Phi_i(x)$.
A first-order expansion of the logarithm at $1$ yields
\eqref{eq:logratio_first_order_snap}.
Finally, by \eqref{eq:gibbs_candidate_snap_unaware},
\[
\log\frac{\widetilde{\pi}_\beta(x)}{\widetilde{\pi}_\beta(x')}
=
\beta \kappa \big(\Phi(x)-\Phi(x')\big),
\]
and adding this identity to \eqref{eq:logratio_first_order_snap} gives
\eqref{eq:approx_detailed_balance_first_order}.
\end{proof}
\begin{remark}
The snapshot chain is not reversible in general, hence $\pi_\beta^{\mathrm{S}}$ typically does not admit
a closed-form Gibbs expression. Theorem~\ref{thm:2} and
Corollary~\ref{cor:gibbs_like_first_order_snap} show that for $\beta\to 0$ the chain is
\emph{Gibbs-like to first order}, with effective potential coefficient $\kappa$.
\end{remark}

\subsection{Fast regime}
\label{subsec:fast}

We next consider a \emph{fast communication} regime in which communication operates on a faster timescale than agent updates.  In this regime, agents respond to \emph{channel-averaged} payoffs,
as if  channel statistics are known.

Agent $i$ uses
\begin{equation}
\widehat u_i^{\mathrm{F}}(b\mid x)
\triangleq
\sum_{j\in N(i)} v_{ij}\,
\mathbb{E}\!\left[\ind\{y_{ij}=b\}\,\middle|\,x_j\right],
\qquad b\in\{0,1\}.
\label{eq:uhat_fast_def}
\end{equation}

A key observation is that for both BSC and BEC, the conditional mean admits a common affine form. This is formalized in the following lemma: 

\begin{lemma}
For both BSC and BEC, there exists a channel-dependent constant $c$ such that
\begin{equation}
\mathbb{E}\!\left[\ind\{y_{ij}=b\}\,\middle|\,x_j\right]
=
c + \kappa\,\ind\{x_j=b\},
\qquad b\in\{0,1\},
\label{eq:channel_affine}
\end{equation}
where $c=p$ for the BSC and $c=0$ for the BEC, and $\kappa$ is as in \eqref{eq:kappa_def}.
\end{lemma}

Substituting \eqref{eq:channel_affine} into \eqref{eq:uhat_fast_def} yields
\begin{equation}
\widehat u_i^{\mathrm{F}}(b\mid x)
=
\underbrace{c\sum_{j\in N(i)} v_{ij}}_{\text{independent of $b$}}
+
\kappa\,m_i(b;x).
\label{eq:uhat_fast_simplified}
\end{equation}

\subsubsection{Fast logit update}

The additive term in \eqref{eq:uhat_fast_simplified} cancels in the logit ratio, and the update becomes
\begin{equation}
\mathbb{P}(x_i^+=b\mid x)
=
\frac{\exp\!\left(\beta\kappa\,m_i(b;x)\right)}
{\sum_{b'\in\{0,1\}}\exp\!\left(\beta\kappa\,m_i(b';x)\right)}.
\label{eq:fast_logit}
\end{equation}

\subsubsection{Exact Gibbs structure}
In the following theorem, we show that the fast communication regime is Gibbs.
\begin{theorem}
\label{thm:fast_gibbs}
Assume $v_{ij}=v_{ji}$. Under the fast update \eqref{eq:fast_logit}, the induced Markov chain is reversible with stationary distribution
\begin{equation}
\pi_\beta^{\mathrm{F}}(x)
=
\frac{1}{Z_\beta^{\mathrm{F}}}\,
\exp\!\left(\beta\kappa\,\Phi(x)\right),
\label{eq:pi_fast}
\end{equation}
where
\begin{equation}
Z_\beta^{\mathrm{F}}
=
\sum_{z\in\{0,1\}^n}
\exp\!\left(\beta\kappa\,\Phi(z)\right).
\label{eq:Z_fast}
\end{equation}
\end{theorem}
\begin{proof}
Fix $i$ and $x_{-i}$. By \eqref{eq:Phi_local_diff},
\begin{equation}
\Phi(b,x_{-i})
=
\Phi(0,x_{-i}) + m_i(b;x)-m_i(0;x).
\label{eq:Phi_cond_app}
\end{equation}
Thus, for the Gibbs distribution \eqref{eq:pi_fast},
\begin{equation}
\pi_\beta^{\mathrm{F}}(x_i=b\mid x_{-i})
\propto
\exp\!\left(\beta\kappa\,m_i(b;x)\right).
\label{eq:gibbs_cond_app}
\end{equation}
Normalizing \eqref{eq:gibbs_cond_app} over $b\in\{0,1\}$ yields the update rule \eqref{eq:fast_logit}.
\end{proof}
\begin{remark}
The channel impacts the fast stationary distribution only through the product $\beta\kappa$.
Thus, communication unreliability acts as an effective temperature increase for fixed $\beta$.
\end{remark}

\section{Finite-$K$ Link Usage: Bridging Snapshot and Fast}
\label{sec:finiteK}

We interpolate between snapshot updates and the fast regime by allowing a finite number
$K\in\mathbb{N}$ of channel uses per neighbor per update.

\subsection{Finite-$K$ empirical payoff}

For each update of agent $i$ and neighbor $j\in N(i)$, assume $K$ i.i.d.\ channel uses:
\begin{equation}
Y_{ij}^{(1)},\dots,Y_{ij}^{(K)}
\ \text{are i.i.d.\ given $x_j$.}
\label{eq:Ksamples_iid}
\end{equation}

For each $b\in\{0,1\}$ define the empirical match frequency
\begin{equation}
\widehat q_{ij}^{(K)}(b)
\triangleq
\frac{1}{K}\sum_{k=1}^K \ind\{Y_{ij}^{(k)}=b\}.
\label{eq:qhat_def}
\end{equation}

Agent $i$ uses the empirical payoff estimate
\begin{equation}
\widehat u_{i}^{(K)}(b)
\triangleq
\sum_{j\in N(i)} v_{ij}\,\widehat q_{ij}^{(K)}(b),
\qquad b\in\{0,1\}.
\label{eq:uhatK_def}
\end{equation}

Define the empirical advantage
\begin{equation}
\widehat\Delta_i^{(K)}
\triangleq
\widehat u_i^{(K)}(1)-\widehat u_i^{(K)}(0).
\label{eq:DeltaHatK}
\end{equation}

\subsection{Finite-$K$ logit update and kernel}

Conditioned on all samples
$Y_i^{1:K}\triangleq \{Y_{ij}^{(k)}: j\in N(i),\,k\le K\}$,
agent $i$ updates via
\begin{equation}
\mathbb{P}(x_i^+=1\mid Y_i^{1:K})
=
\sigmoid\!\left(\beta\,\widehat\Delta_i^{(K)}\right).
\label{eq:logit_finiteK}
\end{equation}

Averaging over the channel and the random node selection yields the induced kernel $P_{\beta,K}$.
For $x_i=0$,
\begin{equation}
P_{\beta,K}\!\left(x,x^{(i,1)}\right)
=
\frac{1}{n}\,
\mathbb{E}\!\left[
\sigmoid\!\left(\beta\,\widehat\Delta_i^{(K)}\right)
\,\middle|\,
x
\right].
\label{eq:P_betaK_forward}
\end{equation}
For $x_i=1$,
\begin{equation}
P_{\beta,K}\!\left(x,x^{(i,0)}\right)
=
\frac{1}{n}\,
\mathbb{E}\!\left[
\sigmoid\!\left(-\beta\,\widehat\Delta_i^{(K)}\right)
\,\middle|\,
x
\right].
\label{eq:P_betaK_backward}
\end{equation}

The following lemma states that $K=1$ scenario here recovers the snapshot communication regime. 
\begin{lemma}
\label{prop:K1_snapshot}
For $K=1$, the kernel $P_{\beta,1}$ coincides with the snapshot kernel
$P_\beta^{\mathrm{S}}$ in \eqref{eq:P_snap_forward}--\eqref{eq:P_snap_backward}.
\end{lemma}
\begin{proof}
If $K=1$ then $$\widehat q_{ij}^{(1)}(b)=\mathbf{1}\{Y_{ij}^{(1)}=b\}$$ and hence
$$\widehat u_i^{(1)}(b)=\sum_{j\in N(i)}v_{ij}\mathbf{1}\{Y_{ij}^{(1)}=b\}=m_i(b;Y_i).$$
Substituting into \eqref{eq:logit_finiteK} yields the snapshot  update.
Averaging over the one-shot channel law gives the stated kernels.
\end{proof}
Define the advantage in the fast communication regime
\begin{equation}
\Delta_i^{\mathrm{F}}(x)
\triangleq
\kappa\,\Delta\Phi_i(x),
\label{eq:Delta_fast_def}
\end{equation}
with $\kappa$ from \eqref{eq:kappa_def}.
The corresponding kernel $P_{\beta,\infty}$  of the fast regime is
\begin{equation}
P_{\beta,\infty}\!\left(x,x^{(i,1)}\right)
=
\frac{1}{n}\,
\sigmoid\!\left(\beta\,\Delta_i^{\mathrm{F}}(x)\right).
\label{eq:P_fast_forward}
\end{equation}
\begin{equation}
P_{\beta,\infty}\!\left(x,x^{(i,0)}\right)
=
\frac{1}{n}\,
\sigmoid\!\left(-\beta\,\Delta_i^{\mathrm{F}}(x)\right).
\label{eq:P_fast_backward}
\end{equation}
In the following we show entrywise kernel convergence $P_{\beta,K}\to P_{\beta,\infty}$. 
\begin{theorem}
\label{thm:P_converges}
Fix $\beta>0$ and either channel model. For every $x\in\{0,1\}^n$ and every $i\in V$,
\begin{equation}
\lim_{K\to\infty} P_{\beta,K}\!\left(x,x^{(i,1)}\right)
=
P_{\beta,\infty}\!\left(x,x^{(i,1)}\right).
\label{eq:entrywise_conv_forward}
\end{equation}
\begin{equation}
\lim_{K\to\infty} P_{\beta,K}\!\left(x,x^{(i,0)}\right)
=
P_{\beta,\infty}\!\left(x,x^{(i,0)}\right).
\label{eq:entrywise_conv_backward}
\end{equation}
\end{theorem}

\begin{proof}
Fix $x$ and $i$. For each neighbor $j$ and action $b$, the indicators $\ind\{Y_{ij}^{(k)}=b\}$ are i.i.d.\ in $k$ and bounded.
By the strong law of large numbers,
\begin{equation}
\widehat q_{ij}^{(K)}(b)
\xrightarrow[K\to\infty]{\text{a.s.}}
\mathbb{E}\!\left[\ind\{Y_{ij}^{(1)}=b\}\,\middle|\,x_j\right].
\label{eq:qhat_slln_app}
\end{equation}
Summing in \eqref{eq:uhatK_def} yields $\widehat u_i^{(K)}(b)\to \widehat u_i^{\mathrm{F}}(b\mid x)$ a.s.,
and hence $\widehat\Delta_i^{(K)}\to \Delta_i^{\mathrm{F}}(x)$ a.s.\ using \eqref{eq:Delta_fast_def}.
Since $\sigmoid(\cdot)$ is bounded and continuous, dominated convergence gives
\begin{equation}
\mathbb{E}\!\left[\sigmoid\!\left(\beta\,\widehat\Delta_i^{(K)}\right)\,\middle|\,x\right]
\to
\sigmoid\!\left(\beta\,\Delta_i^{\mathrm{F}}(x)\right).
\label{eq:domconv_app}
\end{equation}
Substituting \eqref{eq:domconv_app} into \eqref{eq:P_betaK_forward} yields \eqref{eq:entrywise_conv_forward}.
The backward case \eqref{eq:entrywise_conv_backward} follows analogously.
\end{proof}

In the following we present our main result of this section: the process is ergodic for all $K$ and as $K\to\infty$ stationary distributions converge to that of the fast regime.

\begin{theorem}
\label{thm:pi_converges}
Fix $\beta>0$ and either channel model. For every $K\in\mathbb{N}$, the Markov chain with transition matrix $P_{\beta,K}$ is irreducible and aperiodic, hence it admits a unique stationary distribution $\pi_{\beta,K}$.
For every $x\in\{0,1\}^n$,
\begin{equation}
\lim_{K\to\infty}\pi_{\beta,K}(x)
=
\pi_\beta^{\mathrm{F}}(x),
\label{eq:pi_entrywise_conv}
\end{equation}
where $\pi_\beta^{\mathrm{F}}$ is given in \eqref{eq:pi_fast}.
\end{theorem}
\begin{proof}
For any realization of the samples $Y_i^{1:K}$, \eqref{eq:logit_finiteK} assigns strictly positive probability
to both actions $b\in\{0,1\}$. Averaging over $Y_i^{1:K}$ preserves strict positivity of each single-site update.
Thus, from any state $x$ and any coordinate $i$, both moves $x\mapsto x^{(i,0)}$ and $x\mapsto x^{(i,1)}$
occur with positive probability (after accounting for the $1/n$ selection factor).
This implies irreducibility (reach any target by single-site changes) and aperiodicity
(each state has a positive self-loop probability). Hence, the Markov chain with transition matrix $P_{\beta,K}$ admits a stationary distribution $\pi_{\beta,K}$.

By the Markov chain tree theorem, for each $K$ and each $x$, we have
\begin{equation}
\pi_{\beta,K}(x)
=
\frac{\tau_{\beta,K}(x)}{\sum_{z}\tau_{\beta,K}(z)},\end{equation}
where \begin{equation}
\tau_{\beta,K}(x)
\triangleq
\sum_{T\in\mathcal{T}_x}\ \prod_{(z\to z')\in T} P_{\beta,K}(z,z'),
\label{eq:tree_piK}
\end{equation}
and $\mathcal{T}_x$ is the set of directed spanning in-trees rooted at $x$
in the transition graph.
For fixed $x$, $\tau_{\beta,K}(x)$ is a finite sum of finite products of entries of $P_{\beta,K}$,
hence it is a continuous function of those entries.
By Theorem~\ref{thm:P_converges}, $P_{\beta,K}\to P_{\beta,\infty}$ entrywise; therefore
$\tau_{\beta,K}(x)\to\tau_{\beta,\infty}(x)$ for every $x$, and similarly for the denominator.
Taking the ratio yields \eqref{eq:pi_entrywise_conv}.
\end{proof}

\begin{corollary}
\label{cor:objective_conv}
Let $f:\{0,1\}^n\to\mathbb{R}$ be any measurable function. Then
$$\lim_{K\to\infty}\mathbb{E}_{\pi_{\beta,K}}[f(X)]
=
\mathbb{E}_{\pi_{\beta,\infty}}[f(X)].$$
In particular, $$\lim_{K\to\infty}\mathbb{E}_{\pi_{\beta,K}}[\Phi(X)]
=
\mathbb{E}_{\pi_{\beta,\infty}}[\Phi(X)].$$
\end{corollary}

\begin{proof}
The state space is finite and $\pi_{\beta,K}(x)\to\pi_{\beta,\infty}(x)$ pointwise by
Theorem~\ref{thm:pi_converges}; thus the convergence of expectations follows by term-by-term convergence
of the finite sum $\sum_x f(x)\pi_{\beta,K}(x)$.
\end{proof}

\begin{remark}
Each update of agent $i$ uses $K|N(i)|$ channel uses.
Proposition~\ref{prop:K1_snapshot} and Theorems~\ref{thm:P_converges}--\ref{thm:pi_converges} formalize that increasing $K$ interpolates between snapshot behavior ($K=1$) and the fast Gibbs sampler ($K\to\infty$).
\end{remark}

\section{Communication-Theoretic Interpretation}

This section provides a communication-theoretic interpretation of the attenuation
coefficients derived in Sections~III and~IV and of the finite-$K$ interpolation between
snapshot and fast communication regimes.
We show that the learning dynamics can be interpreted through classical concepts such as
repetition coding, decoding, and estimation over noisy channels.

\subsection{Repetition Coding over a Binary Symmetric Channel}

Consider a directed link $j\to i$ modeled as a binary symmetric channel (BSC) with crossover
probability $p\in[0,\tfrac12)$.
Suppose agent $j$ transmits its action $x_j\in\{0,1\}$ to agent $i$ using $K$ independent
channel uses.

\subsubsection{Decoding-Based Aggregation}

If agent $i$ performs decoding based on the $K$ received symbols and then applies the
learning rule to the decoded action, the natural decoder under a symmetric prior is
majority vote (maximum-likelihood decoding).
Let $\hat x_j$ denote the decoded bit.
The effective error probability after decoding is
\[
p_K
=
\sum_{m=\lceil(K+1)/2\rceil}^{K}
\binom{K}{m} p^m (1-p)^{K-m}.
\]
Thus repetition coding converts the original BSC$(p)$ into an effective BSC$(p_K)$.

Substituting $p_K$ into the attenuation factor of the fast regime yields
\[
\kappa_K = 1 - 2p_K.
\]
As $K\to\infty$, $p_K\to 0$ and hence $\kappa_K\to 1$, recovering the noiseless coordination
game.
This interpretation explains how sufficient retransmissions eliminate channel noise at the
cost of increased communication.

\subsubsection{Estimation-Based Aggregation}

Alternatively, agent $i$ may avoid hard decoding and instead form an estimate of the
\emph{expected match payoff} directly from the received samples.
Specifically, let
\[
\hat q_{ij}^{(K)}(b)
=
\frac{1}{K}\sum_{k=1}^{K}\mathbf{1}\{Y_{ij}^{(k)}=b\},
\]
and construct the estimated payoff
\[
\hat u_i^{(K)}(b)
=
\sum_{j\in N(i)} v_{ij}\,\hat q_{ij}^{(K)}(b).
\]
This estimator converges almost surely to the channel-averaged payoff as $K\to\infty$.
In contrast to decoding-based aggregation, the mean of this estimator is \emph{exactly}
attenuated by $(1-2p)$ for all $K$, while its variance decreases with $K$.

This distinction explains the behavior observed in Section~IV:
finite-$K$ learning interpolates smoothly between snapshot communication ($K=1$) and the
fast regime ($K\to\infty$), even though the attenuation coefficient itself does not change
with $K$ under estimation-based aggregation.

\subsection{Binary Erasure Channel Interpretation}

The same interpretation applies to a binary erasure channel (BEC) with erasure probability
$\epsilon\in[0,1)$.
Each transmission reveals the correct bit with probability $1-\epsilon$ and provides no
information otherwise.

Under estimation-based aggregation, the expected match payoff is attenuated by a factor
$(1-\epsilon)$, yielding the effective coefficient
\[
\kappa = 1 - \epsilon.
\]
As in the BSC case, increasing $K$ reduces variance but does not alter the mean attenuation
unless explicit decoding is performed.
This explains why finite-$K$ communication improves performance monotonically but with
diminishing returns, as observed in the numerical results.

\subsection{Interpretation of the Two Communication Regimes}

The communication-theoretic viewpoint clarifies the qualitative difference between snapshot
and fast learning dynamics.
Snapshot communication corresponds to acting on a \emph{single realization} of the noisy
channel output.
Fast communication corresponds to acting on the \emph{expectation} of the channel output,
either through sufficiently fast signaling or through averaging over many channel uses.

From an information-theoretic perspective, the fast regime suppresses channel noise before
the nonlinear logit update, whereas the snapshot regime applies the nonlinearity prior to
averaging.
This non-commutativity explains both the Gibbs structure of the fast regime and the
non-equilibrium behavior of snapshot dynamics.

\subsection{Implications for Communication-Constrained Learning}

The above interpretation highlights a fundamental tradeoff between communication cost and
learning performance.
Finite-$K$ communication provides a continuous design space between minimal signaling and
full channel averaging.
In practice, moderate values of $K$ may achieve near-fast performance while significantly
reducing communication overhead, as demonstrated in Section~VIII.

More broadly, the results suggest that the manner in which communication noise is processed
is as important as the noise level itself.
Expectation-based aggregation yields equilibrium learning dynamics with predictable
behavior, while realization-based aggregation induces non-equilibrium effects that cannot
be mitigated solely by increasing computational rationality.

\section{Fast Regime as Entropy-Regularized Optimization}
\label{sec:variational}

This section interprets the fast regime (Theorem~\ref{thm:fast_gibbs}) as the optimizer of an entropy-regularized objective over distributions on $\{0,1\}^n$. The results in this section are trivial extensions of known properties of Gibbs distribution, hence we relegate the proofs to appendix. The key observation  here is that under fast communication regime, the noisy learning dynamics is equivalent to noiseless dynamics with effective temperature $\beta_{eff}=\kappa \beta$. 

Let $\Delta$ denote the probability simplex over $\{0,1\}^n$.
For $\mu\in\Delta$, define Shannon entropy
\begin{equation}
H(\mu)
\triangleq
-\sum_{x\in\{0,1\}^n}\mu(x)\log \mu(x),
\label{eq:entropy_def}
\end{equation}
with the convention $0\log 0=0$.

Define the Lagrangian (free-energy) functional
\begin{equation}
\mathcal{J}_{\beta}(\mu)
\triangleq
\kappa\,\mathbb{E}_{\mu}[\Phi(X)]
+
\frac{1}{\beta}H(\mu),
\label{eq:J_def}
\end{equation}
where $\kappa$ is defined in \eqref{eq:kappa_def} and $\Phi$ is defined in \eqref{eq:Phi_def}.
The following theorem states that the Gibbs distribution is the unique maximizer of the Lagrangian above. 
\begin{theorem}
\label{thm:variational_gibbs}
Fix $\beta>0$.
The fast stationary distribution $\pi_\beta^{\mathrm{F}}$ in \eqref{eq:pi_fast} is the unique maximizer
\begin{equation}
\pi_\beta^{\mathrm{F}}
=
\arg\max_{\mu\in\Delta}\ \mathcal{J}_{\beta}(\mu).
\label{eq:argmax_gibbs}
\end{equation}
Moreover,
\begin{equation}
\max_{\mu\in\Delta}\mathcal{J}_{\beta}(\mu)
=
\frac{1}{\beta}\log Z_\beta^{\mathrm{F}},
\label{eq:opt_value_logZ}
\end{equation}
where $Z_\beta^{\mathrm{F}}$ is defined in \eqref{eq:Z_fast}.
\end{theorem}
The following theorem characterizes the optimization gap:  
\begin{theorem}
\label{thm:gap_bound}
Let $\Phi^\star\triangleq \max_x \Phi(x)$.
Under $\pi_\beta^{\mathrm{F}}$,
\begin{equation}
\Phi^\star - \mathbb{E}_{\pi_\beta^{\mathrm{F}}}[\Phi(X)]
\le
\frac{n\log 2}{\beta\kappa}.
\label{eq:gap_bound_main}
\end{equation}
\end{theorem}
\begin{remark}[Communication reliability as effective temperature]
By \eqref{eq:gap_bound_main}, channel imperfections degrade performance only through the scalar $\kappa$.
For BSC$(p)$, $\kappa=1-2p$; for BEC$(\epsilon)$, $\kappa=1-\epsilon$.
Thus, for a fixed physical temperature $1/\beta$, a less reliable channel reduces $\beta\kappa$
and increases the optimization gap bound proportionally.
\end{remark}

\section{Heterogeneous Link Reliabilities}
\label{sec:hetero}

We extend the model to link-dependent channel qualities, as in wireless networks with heterogeneous SNR.

Let $G=(V,E)$ be an undirected graph. To each undirected edge $\{i,j\}\in E$ we associate an edge-dependent channel parameter and impose direction-independence (symmetry) of the channel quality. More formally, we have the following symmetry assumptions throughout this section.

\begin{assumption}\label{ass:heterogeneous_bsc}
In the heterogeneous binary symmetric channel (BSC) setting, each edge $\{i,j\}\in E$ has a crossover probability
\[
p_{ij}\in\Bigl[0,\tfrac{1}{2}\Bigr), \qquad p_{ij}=p_{ji}.
\]
\end{assumption}

\begin{assumption}\label{ass:heterogeneous_bec}
In the heterogeneous binary erasure channel (BEC) setting, each edge $\{i,j\}\in E$ has an erasure probability
\[
\epsilon_{ij}\in[0,1), \qquad \epsilon_{ij}=\epsilon_{ji}.
\]
\end{assumption}

Define the \emph{edge reliability coefficient} $\kappa_{ij}$ for every $\{i,j\}\in E$ by
\begin{equation}
\kappa_{ij}\triangleq
\begin{cases}
1-2p_{ij}, & \text{(heterogeneous BSC)},\\[2mm]
1-\epsilon_{ij}, & \text{(heterogeneous BEC)}.
\end{cases}
\label{eq:kappa_ij}
\end{equation}
Note that under the assumptions above,  $\kappa_{ij}=\kappa_{ji}$ holds for all $\{i,j\}\in E$.


\begin{figure*}
\centering
\begin{minipage}[b]{0.23\linewidth}
  \centering
     \centerline{\includegraphics[width=4.5cm]{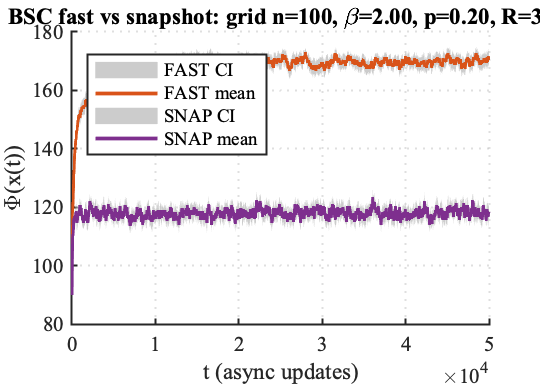}}
  \centerline{\scriptsize (a) $p=0.2$, grid network, $\beta=2$.}
\end{minipage} 
\hfill 
\begin{minipage}[b]{0.23\linewidth}
  \centering
     \centerline{\includegraphics[width=4.5cm]{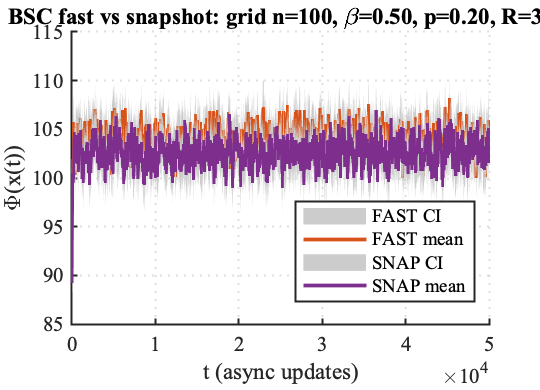}}
  \centerline{ \scriptsize (b) $p=0.2$, grid network, $\beta=0.5$.}
\end{minipage}
\hfill 
\begin{minipage}[b]{0.23\linewidth}
  \centering
     \centerline{\includegraphics[width=4.5cm]{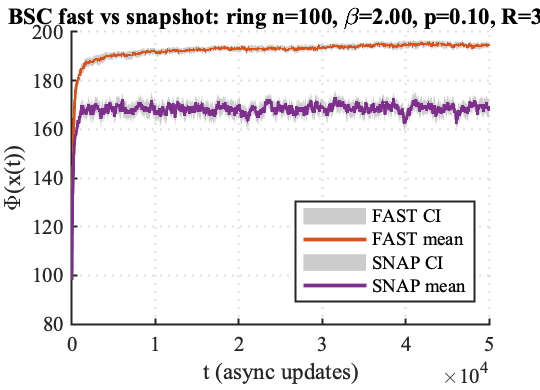}}
  \centerline{ \scriptsize (c) $p=0.1$, ring network, $\beta=2$.}
\end{minipage}
\hfill 
\begin{minipage}[b]{0.23\linewidth}
  \centering
     \centerline{\includegraphics[width=4.5cm]{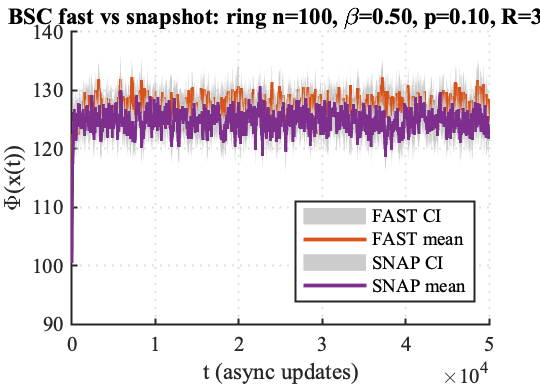}}
  \centerline{ \scriptsize (d) $p=0.1$, ring network, $\beta=0.5$.}
\end{minipage}
\caption{Means and 95\% confidence intervals of the steady-state coordination potential for snapshot and fast communication over BSC links across different network topologies. In the high-temperature regime (
$\beta=0.5$) the two dynamics exhibit comparable performance, as predicted by the first-order approximation. At lower temperatures ($\beta=2$) fast communication significantly outperforms snapshot communication and exhibits reduced variability, owing to channel-averaged updates, whereas snapshot communication displays higher variance due to its dependence on single-shot channel observations.}
\label{fig:quantizers1}
\end{figure*}


Define effective weights
\begin{equation}
w_{ij}
\triangleq
v_{ij}\kappa_{ij}.
\label{eq:w_ij_def}
\end{equation}

Define the effective potential
\begin{equation}
\Phi_{\mathrm{eff}}(x)
\triangleq
\sum_{\{i,j\}\in E} w_{ij}\,\ind\{x_i=x_j\}.
\label{eq:Phi_eff}
\end{equation}
The following theorem states that under the heterogeneous channel model,  fast communication regime dynamics remain to be Gibbs. 
\begin{theorem}
\label{thm:hetero_fast_gibbs}
Assume $v_{ij}=v_{ji}$ and $\kappa_{ij}=\kappa_{ji}$ for all $\{i,j\}\in E$.
Under fast communication, asynchronous logit updates generate a reversible Markov chain with stationary distribution
\begin{equation}
\pi_{\beta}^{\mathrm{F,het}}(x)
=
\frac{1}{Z_{\beta}^{\mathrm{F,het}}}\,
\exp\!\left(\beta\,\Phi_{\mathrm{eff}}(x)\right),
\label{eq:pi_fast_het}
\end{equation}
where
\begin{equation}
Z_{\beta}^{\mathrm{F,het}}
=
\sum_{z\in\{0,1\}^n}\exp\!\left(\beta\,\Phi_{\mathrm{eff}}(z)\right).
\label{eq:Z_fast_het}
\end{equation}
\end{theorem}

We next focus on the high temperature approximation of the snaphot dynamics. 

Define the effective local match count
\begin{equation}
m_i^{\mathrm{eff}}(b;x)
\triangleq
\sum_{j\in N(i)} v_{ij}\kappa_{ij}\,\ind\{x_j=b\},
\qquad b\in\{0,1\}.
\label{eq:mi_eff}
\end{equation}
The following theorem states that under the heterogeneous channel model,  snapshot communication regime dynamics can be approximated (in the first order) as Gibbs at high temperature. 
\begin{theorem}
\label{thm:small_beta_het}
Fix $x\in\{0,1\}^n$ and $i\in V$.
As $\beta\to 0$, for $x_i=0$,
\begin{equation}
P_{\beta}^{\mathrm{S,het}}\!\left(x,x^{(i,1)}\right)
=
\frac{1}{2n}
+
\frac{\beta}{4n}
\Big(m_i^{\mathrm{eff}}(1;x)-m_i^{\mathrm{eff}}(0;x)\Big)
+
o(\beta).
\label{eq:small_beta_het_forward}
\end{equation}
\end{theorem}

\begin{remark}
Heterogeneity in the fast regime simply rescales edges via $w_{ij}$ in \eqref{eq:w_ij_def}.
In the snapshot regime, the dynamics remain nonreversible in general, but their high-temperature drift is governed by the same effective local field.
\end{remark}
\begin{figure*}
\centering
\begin{minipage}[b]{0.23\linewidth}
  \centering
     \centerline{\includegraphics[width=4.5cm]{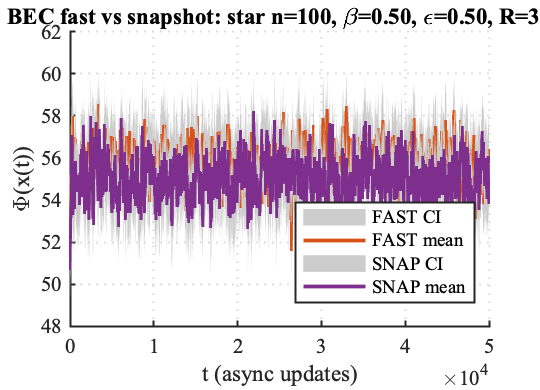}}
  \centerline{\scriptsize (a) $\varepsilon=0.5$, star network, $\beta=0.5$.}
\end{minipage} 
\hfill 
\begin{minipage}[b]{0.23\linewidth}
  \centering
     \centerline{\includegraphics[width=4.5cm]{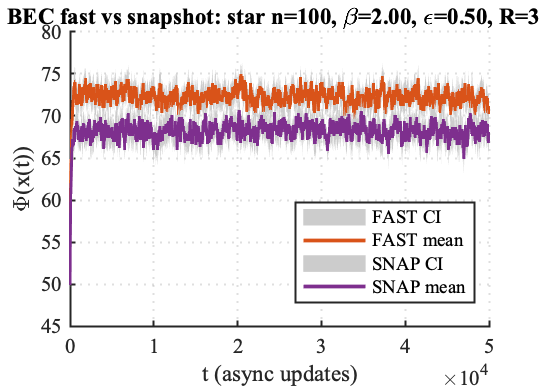}}
  \centerline{ \scriptsize (b) $\varepsilon=0.5$, star network, $\beta=2$.}
\end{minipage}
\hfill 
\begin{minipage}[b]{0.23\linewidth}
  \centering
     \centerline{\includegraphics[width=4.5cm]{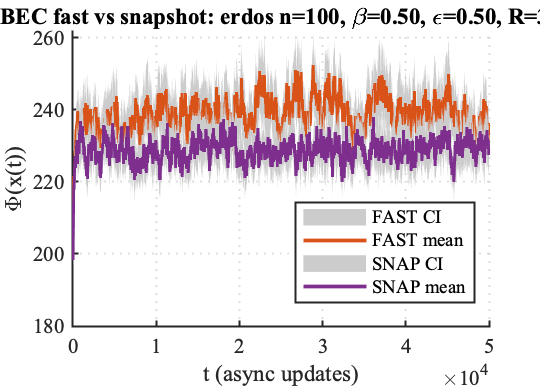}}
  \centerline{ \scriptsize (c) $\varepsilon=0.5$, erdos network, $\beta=0.5$.}
\end{minipage}
\hfill 
\begin{minipage}[b]{0.23\linewidth}
  \centering
     \centerline{\includegraphics[width=4.5cm]{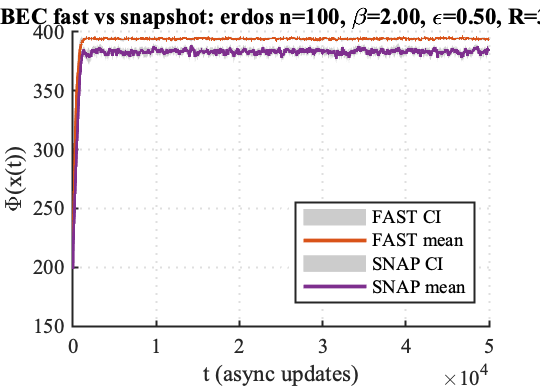}}
  \centerline{ \scriptsize (d) $\varepsilon=0.5$, erdos network, $\beta=2$.}
\end{minipage}
\caption{Means and 95\% confidence intervals of the steady-state coordination potential for snapshot and fast communication over BEC links across different network topologies. In the high-temperature regime (
$\beta=0.5$) the two dynamics exhibit comparable performance, as predicted by the first-order approximation. At lower temperatures ($\beta=2$) fast communication significantly outperforms snapshot communication and exhibits reduced variability, owing to channel-averaged updates, whereas snapshot communication displays higher variance due to its dependence on single-shot channel observations.}
\label{fig:quantizers1}
\end{figure*}

\section{Numerical Results}

This section illustrates the theoretical results through numerical experiments.
Our objectives are threefold: 
(i) to validate the separation between snapshot and fast communication regimes,
(ii) to demonstrate the interpolation induced by a finite communication budget $K$,
and (iii) to quantify the impact of heterogeneous link reliabilities.
All experiments simulate the actual asynchronous log-linear learning dynamics described
in Sections~III–V.

\subsection{Experimental Setup}

We consider the following representative network topologies:
\begin{itemize}
\item \emph{Ring}: each node connected to its two nearest neighbors;
\item \emph{Grid}: a $10\times10$ two-dimensional lattice;
\item \emph{Erd\H{o}s--R\'enyi}: random graph $G(n,q)$ with $q=0.08$;
\item \emph{Star}: one central node connected to all others.
\end{itemize}
Unless stated otherwise, all edges have unit weights $v_{ij}=1$.
All graphs have $n=100$ nodes.

At each discrete time step, a single agent is selected uniformly at random and updates
its action according to the logit rule with inverse temperature $\beta$.
In the snapshot regime, each update uses one noisy observation per incident edge.
In the fast regime, agents update using channel-averaged payoffs.
For the finite-$K$ regime, each update uses $K$ independent channel uses per neighbor.

For BSC experiments we use
\[
p \in \{0,\,0.05,\,0.1,\,0.2\},
\]
and for BEC experiments
\[
\epsilon \in \{0,\,0.1,\,0.3,\,0.5\}.
\]
Unless otherwise specified, $\beta=2$.

Each simulation is run for a sufficiently long horizon to ensure convergence,
with an initial burn-in period discarded.
Steady-state performance is estimated by time-averaging the potential $\Phi(x(t))$.
All reported results are averaged over $25$ independent runs with different random seeds.
Shaded regions indicate $95\%$ confidence intervals.

\subsection{Fast vs.\ Snapshot Communication}

Figures~1 and~2 show the steady-state coordination potential under snapshot and fast
communication for BSC and BEC links, respectively, across different network topologies.

For small inverse temperature ($\beta=0.5$), the two regimes exhibit nearly identical
mean performance.
This behavior is consistent with the high-temperature expansion derived in
Theorems~2 and~9, which shows that the drift of snapshot dynamics matches that of the
fast Gibbs sampler to first order in $\beta$.

As the inverse temperature increases ($\beta=2$), a clear separation emerges.
Across all topologies and both channel models, fast communication consistently achieves
higher steady-state potential than snapshot communication.
Moreover, the fast regime exhibits significantly reduced variability across runs.
This reduced dispersion is a direct consequence of channel-averaged updates,
which suppress randomness due to instantaneous channel realizations.
In contrast, snapshot communication relies on single-shot observations and therefore
exhibits higher variance and less stable convergence.

These results empirically confirm the qualitative distinction between equilibrium
(Gibbs) behavior in the fast regime and non-equilibrium behavior in the snapshot regime.

\subsection{Finite-$K$ Communication Budget}
Figure~3 illustrates the effect of a finite communication budget under BSC links.
The parameter $K$ denotes the number of independent channel uses per neighbor and update.
The fast regime is shown as a reference line.

As $K$ increases, finite-$K$ communication progressively averages out channel noise and
the resulting steady-state performance converges monotonically toward the fast regime.
The convergence exhibits diminishing returns: moderate values of $K$ (typically
$K\approx5$–$10$) already capture most of the performance gain relative to snapshot
communication.

This behavior aligns with the kernel convergence and stationary distribution convergence
results in Theorems~4 and~5.
From a communication-theoretic perspective, finite-$K$ communication corresponds to
retransmissions or repetition coding, providing a concrete tradeoff between communication
resources and coordination quality.

\subsection{Heterogeneous Link Reliabilities}

Figure~4 examines coordination performance under heterogeneous BSC links, where edge
reliabilities are drawn from a specified range.
In the fast regime, the dynamics behave as predicted by the effective potential
$\Phi_{\mathrm{eff}}$ in~(68): heterogeneous reliabilities simply rescale edge weights
and the learning dynamics remain Gibbs.

In contrast, snapshot communication is more sensitive to unreliable edges.
As link heterogeneity increases, the performance gap between fast and snapshot
communication widens.
This effect is particularly pronounced in topologies with hub nodes, such as star
networks, where a single poorly conditioned edge can disproportionately influence local
updates.

These observations are consistent with Theorem~9 and Remark~5, which show that while the
snapshot regime remains non-reversible, its high-temperature drift is governed by the
same effective local field as the fast regime.

\subsection{Discussion}

Taken together, the numerical results support the main theoretical conclusions of the
paper.
Fast communication yields equilibrium behavior characterized by a Gibbs distribution,
with reduced variability and predictable dependence on channel reliability.
Snapshot communication, while simpler, induces a non-equilibrium process whose
performance degrades and becomes more variable as the system moves away from the
high-temperature regime.
The finite-$K$ model provides a principled interpolation between these extremes and
highlights the role of communication resources in shaping distributed learning outcomes.

\begin{figure}[htbp]
\centering
\begin{minipage}[b]{0.48\linewidth}
     \centerline{\includegraphics[width=4.1cm]{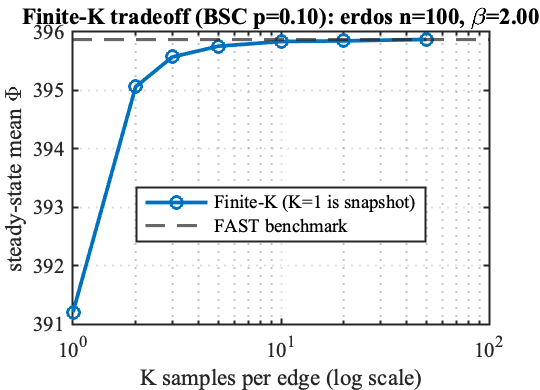}}
  \centerline{\scriptsize (a) BSC $p=0.1$, Erdos network, $\beta=2$.}
\end{minipage}  
\hfill
\begin{minipage}[b]{0.48\linewidth}
     \centerline{\includegraphics[width=4.1cm]{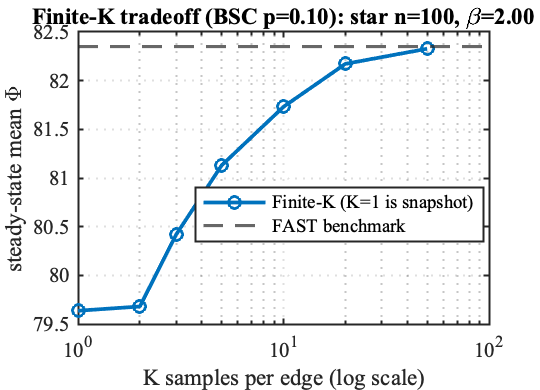}}
  \centerline{ \scriptsize (b) BSC $p=0.1$, star network, $\beta=2$.}
\end{minipage}
\caption{Steady-state coordination performance under finite-$K$ communication over BSC links across different network topologies, with the fast (channel-averaged) regime shown as a reference line. The parameter $K$ denotes the number of independent channel uses per neighbor and update. As $K$ increases, finite-$K$ communication progressively averages out channel noise and the resulting performance converges toward the fast regime. The improvement is monotone but exhibits diminishing returns, indicating that moderate values of $K$ capture most of the benefit of expectation-based updates.}
\label{fig:3}
\end{figure}

\begin{figure}[htbp]
\centering
\begin{minipage}[b]{0.48\linewidth}
     \centerline{\includegraphics[width=4.1cm]{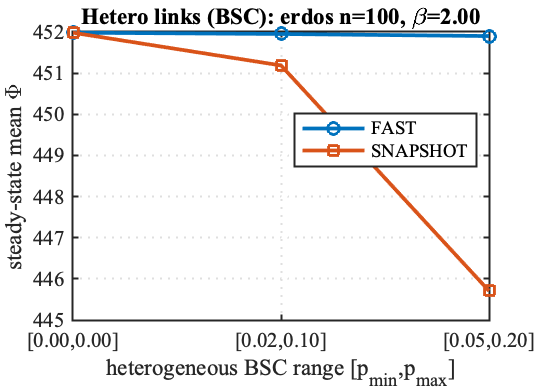}}
  \centerline{\scriptsize (a) Erdos network, $\beta=2$.}
\end{minipage}  
\hfill
\begin{minipage}[b]{0.48\linewidth}
     \centerline{\includegraphics[width=4.1cm]{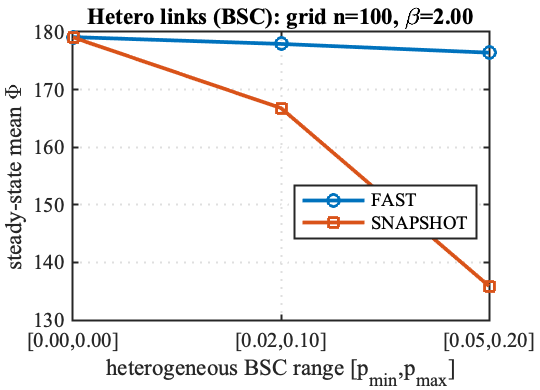}}
  \centerline{ \scriptsize (b) grid network, $\beta=2$.}
\end{minipage}
\caption{Steady-state coordination performance under heterogeneous BSC links, where edge reliabilities are drawn from a specified range. As link heterogeneity increases, the performance gap between fast and snapshot communication widens. Fast communication mitigates the effect of unreliable edges through expectation-based updates, while snapshot communication exhibits greater sensitivity to link variability.}
\label{fig:3}
\end{figure}

\section{Conclusion}

This paper studied distributed log-linear learning under noisy communication, with a
particular focus on how communication primitives fundamentally shape the long-run behavior
of learning dynamics.
We distinguished between two canonical regimes: \emph{snapshot communication}, in which
agents act on one-shot noisy observations, and \emph{fast communication}, in which agents
update using channel-averaged information.

Our analysis showed that these two regimes induce qualitatively different stochastic
processes.
In the fast regime, the learning dynamics reduce exactly to a Gibbs sampler associated with
a scaled coordination potential, yielding a reversible Markov chain with a closed-form
stationary distribution.
Channel noise affects the system only through a single reliability parameter, resulting in
a precise temperature-scaling interpretation.
In contrast, snapshot communication induces a non-reversible, non-equilibrium Markov
process whose stationary distribution admits no closed-form expression.
Nevertheless, we established that snapshot dynamics are \emph{Gibbs-like} in the
high-temperature regime, providing a first-order analytical link between the two settings.

We further introduced a finite-$K$ communication model that interpolates between snapshot
and fast communication.
This model captures practical communication constraints such as retransmissions or limited
averaging and yields a smooth performance transition as the communication budget increases.
Both theoretical results and numerical experiments demonstrate diminishing returns in $K$,
highlighting a fundamental tradeoff between communication cost and coordination
performance.

Numerical results across multiple network topologies and channel models validated the
theoretical predictions.
Fast communication consistently achieved higher steady-state performance and exhibited
lower variability due to its expectation-based updates, while snapshot communication was
more sensitive to instantaneous channel realizations and heterogeneous link reliabilities.
Exact computations for small networks and large-scale simulations together provided a
complete characterization of steady-state behavior across regimes.

From a broader perspective, the results clarify when distributed learning dynamics can be
interpreted as equilibrium sampling procedures and when they instead operate far from
equilibrium due to communication constraints.
These distinctions are particularly relevant for distributed optimization and learning
systems operating over unreliable or bandwidth-limited networks.

Several directions for future work are suggested by this study.
Extensions to multi-action games, asymmetric or time-varying channels, adaptive
communication budgets, and coupling with learning-rate or temperature schedules are
natural next steps.
More broadly, the framework developed here provides a principled foundation for analyzing
distributed learning algorithms at the intersection of game theory, information theory,
and networked control.

\bibliographystyle{IEEEtran}
\bibliography{refs}

\appendices

\section{Proofs for Section~\ref{sec:variational}}
\label{app:variational_proofs}

\subsection{Proof of Theorem~\ref{thm:variational_gibbs}}

Using \eqref{eq:pi_fast}, write
\begin{equation}
\log \pi_\beta^{\mathrm{F}}(x)
=
\beta\kappa\,\Phi(x)-\log Z_\beta^{\mathrm{F}}.
\label{eq:log_pi_app}
\end{equation}
Consider the KL divergence $D(\mu\Vert\pi_\beta^{\mathrm{F}})\ge 0$,
\begin{equation}
D(\mu\Vert\pi_\beta^{\mathrm{F}})
\triangleq
\sum_x \mu(x)\log\frac{\mu(x)}{\pi_\beta^{\mathrm{F}}(x)}.
\label{eq:KL_def_app}
\end{equation}
Expanding and substituting \eqref{eq:log_pi_app} gives
\begin{equation}
D(\mu\Vert\pi_\beta^{\mathrm{F}})
=
-H(\mu)
-\beta\kappa\,\mathbb{E}_\mu[\Phi(X)]
+\log Z_\beta^{\mathrm{F}}.
\label{eq:KL_expand_app}
\end{equation}
Rearranging \eqref{eq:KL_expand_app} yields
\begin{equation}
\mathcal{J}_{\beta}(\mu)
=
\frac{1}{\beta}\log Z_\beta^{\mathrm{F}}
-\frac{1}{\beta}D(\mu\Vert\pi_\beta^{\mathrm{F}})
\le
\frac{1}{\beta}\log Z_\beta^{\mathrm{F}}.
\label{eq:J_upper_app}
\end{equation}
Equality holds iff $D(\mu\Vert\pi_\beta^{\mathrm{F}})=0$, i.e., $\mu=\pi_\beta^{\mathrm{F}}$.
This proves \eqref{eq:argmax_gibbs} and \eqref{eq:opt_value_logZ}.

\subsection{Proof of Theorem~\ref{thm:gap_bound}}

Let $x^\star\in\arg\max_x \Phi(x)$ and let $\delta_{x^\star}$ be the point mass at $x^\star$.
Then $H(\delta_{x^\star})=0$ and $\mathbb{E}_{\delta_{x^\star}}[\Phi(X)]=\Phi^\star$.
By optimality of $\pi_\beta^{\mathrm{F}}$,
\begin{equation}
\mathcal{J}_{\beta}(\pi_\beta^{\mathrm{F}})
\ge
\mathcal{J}_{\beta}(\delta_{x^\star})
=
\kappa\,\Phi^\star.
\label{eq:J_compare_app}
\end{equation}
Expanding and rearranging yields
\begin{equation}
\Phi^\star-\mathbb{E}_{\pi_\beta^{\mathrm{F}}}[\Phi(X)]
\le
\frac{1}{\beta\kappa}\,H(\pi_\beta^{\mathrm{F}}).
\label{eq:gap_via_entropy_app}
\end{equation}
Since $\pi_\beta^{\mathrm{F}}$ is supported on $\{0,1\}^n$, $H(\pi_\beta^{\mathrm{F}})\le n\log 2$.
Substituting gives \eqref{eq:gap_bound_main}.

\end{document}